\documentclass[acmsmall,screen,nonacm]{acmart}

\usepackage{url}
\usepackage{xspace}
\usepackage{wrapfig}
\usepackage{hyperxmp}
\usepackage{booktabs}
\usepackage{makecell}
\usepackage{hyperref}
\usepackage{graphicx}
\usepackage{flafter}
\usepackage{colortbl}
\usepackage{multirow}
\usepackage{enumitem}
\usepackage{subfigure}
\usepackage{todonotes}
\usepackage{tablefootnote}
\floatstyle{plaintop}
\newfloat{interactioncontract}{tbp}{loc}
\floatname{interactioncontract}{Contract}
\newcommand{\contractref}[1]{Contract~\ref{#1}\xspace}
\usepackage{color, xcolor}
\usepackage{amsthm,amsmath,amsfonts}
\newcommand{\ie}{\textit{i.e.,}\xspace}
\newcommand{\eg}{\textit{e.g.,}\xspace}
\newcommand{\etc}{\textit{etc.}\xspace}
\newcommand{\etal}{\textit{et al.}\xspace}

\newcommand{\figref}[1]{Fig.~\ref{#1}\xspace}
\newcommand{\tabref}[1]{Table~\ref{#1}\xspace}
\newcommand{\secref}[1]{Section~\ref{#1}\xspace}

\newcommand{\toolname}{IIA\xspace}
\newcommand{\aisname}{AIS\xspace}

\setcopyright{none}
\acmDOI{}
\acmArticle{}
\date{10 September 2026}

\theoremstyle{acmdefinition}
\newtheorem{definition}{Definition}
\theoremstyle{acmplain}

\begin{document}

\title[Engineering Agent-Integrated Software]{Engineering Agent-Integrated Software: Interaction Contracts and Continuous Assurance}

\author{Shengcheng Yu}
\affiliation{\institution{Technical University of Munich}\city{Heilbronn}\country{Germany}\postcode{74076}}
\email{shengcheng.yu@tum.de}
\orcid{0000-0003-4640-8637}

\author{Chunrong Fang}
\authornote{Chunrong Fang is the corresponding author.}
\affiliation{\institution{State Key Laboratory for Novel Software Technology, Nanjing University}\city{Nanjing}\country{China}\postcode{210093}}
\email{fangchunrong@nju.edu.cn}
\orcid{0000-0002-9930-7111}

\author{Zhenyu Chen}
\affiliation{\institution{State Key Laboratory for Novel Software Technology, Nanjing University}\city{Nanjing}\country{China}\postcode{210093}}
\email{zychen@nju.edu.cn}
\orcid{0000-0002-9592-7022}

% \authornote{ is the corresponding author.}

\begin{abstract}

Embedding an intelligent agent in an existing application creates a persistent coordination problem: users can revise goals and manipulate shared objects while delegated execution continues. We argue that dependable integration requires an explicit correspondence between task-level interaction and application behavior. We introduce Agent-Integrated Software (AIS) as a software pattern combining a conventional core, direct interaction, and a built-in agent, and Intent-Level Interaction Abstraction (IIA) as the task semantics through which users inspect and control delegated work. An open transition-system model relates AIS execution to IIA states and events. Interaction contracts constrain this relation through task bindings, role-specific authority, control transitions, and outcome evidence; continuous assurance maintains scoped claims as their dependencies change. A compact disclosure contract and conditional propositions illustrate why local component validity is insufficient and how selected admission invariants can be separated from planning. Contrasting software domains expose the framework's assumptions and limits. This perspective develops a research agenda spanning application abstraction, development support, controlled execution, quality assessment, and human supervision, with the aim of making agent integration a maintainable software engineering discipline.

\end{abstract}




\maketitle

\section{Introduction}
\label{sec:introduction}

Embedding a language-model agent in an application changes the relationship between interaction and execution. A user can delegate a goal while continuing to edit the very objects on which the agent is acting. The application must then interpret two kinds of input: direct operations on its existing interface and instructions whose realization may require a sequence of inferred operations. Their coexistence creates a software engineering problem that cannot be assessed through the quality of the generated response alone. A convincing proposal, an authorized application programming interface (API) call, and a correct database update can still compose into an effect that no longer corresponds to the user's instruction.

Consider a collaboration application in which an organizer asks an agent to prepare meeting materials and send public versions to confirmed participants. The organizer inspects the proposal, replaces an attachment, and removes a recipient through the conventional graphical user interface (GUI). A delayed approval for the earlier proposal may remain authentic even though its scope is obsolete. If a delivery subsequently times out, the absence of a response does not establish whether disclosure occurred. This example brings the central difficulty into focus: the meaning of delegated work can change before, during, and after execution, while different people retain authority over the resources and consequences involved.

\label{sec:definition}
We introduce \emph{Agent-Integrated Software} (\aisname) as a software-pattern concept for this setting. \aisname integrates a conventional application core with a built-in intelligent agent, supporting both direct user interaction and goal-directed task execution. The core continues to own the application's domain objects, business behavior, and durable state. The agent is integrated into the application's functionality and lifecycle, although its model or runtime may execute remotely. The integrated agent uses a large language model or multimodal foundation model. Existing applications can adopt the pattern incrementally.

Within \aisname, the \emph{Intent-Level Interaction Abstraction} (\toolname) gives delegated work a user-facing semantics: goals, contextual references, proposals, endorsements, interventions, and outcomes. It is the abstraction through which users inspect and redirect a task, \ie the meaning of the interaction rather than the component that performs inference. Its realization can span a conversational panel, an editable preview, existing GUI controls, and runtime services. The agent plans and executes; the \toolname defines how that activity is presented and controlled; \aisname identifies the application in which both interaction paths coexist.

\label{sec:conceptual-status}
Our contribution is to make the relationship between interaction and application effects an explicit object of specification. We call this relationship an \emph{interaction--effect obligation}. It binds an effect to the task revision, referenced objects, role-specific authority, operational controller, and evidence that justifies its reported outcome. Mixed-initiative interaction and co-planning already address the distribution of initiative and revision of shared plans~\cite{horvitz1999mixed,feng2026cocoa}; intent-oriented software and agent--UI protocols address complementary software and communication concerns~\cite{xie2026intent,agui2026protocol}. \aisname and \toolname introduce an engineering boundary and semantic abstraction for the continuing application. Their contribution is the joint treatment of obligations across that boundary, independently of a particular protocol, access-control scheme, or transaction mechanism.

\label{sec:position}
\textbf{Our position is that dependable agent integration requires a maintained correspondence between revisable task intent and application behavior.} This correspondence should guide development support and quality assessment from the outset. We formulate \aisname as an open transition system and \toolname as an abstraction of its task-relevant behavior. Interaction contracts constrain the correspondence; continuous assurance records which parts remain supported as tasks, policies, and implementations change. The formulation separates architectural membership from dependability: an application can instantiate \aisname and still violate its interaction contract.

This perspective develops the argument from the semantic framework in \secref{sec:framework}, through its engineering implications in \secref{sec:challenges}, to a compact contract and assurance model in \secref{sec:quality}. The resulting research agenda connects application analysis, interface design, controlled execution, and quality assessment. Contrasting examples in collaboration, spreadsheets, integrated development environments (IDEs), and service consoles examine the scope of the argument. The contribution is a conceptual framework with conditional reasoning and worked examples; its practical benefits remain questions for comparative engineering and human studies.

\section{AIS and IIA: A Semantic Framework}
\label{sec:framework}

An architectural diagram identifies components, but it does not determine whether a user's intervention changes the behavior that those components can produce. We therefore describe \aisname at two levels: application execution and task interaction. The distinction lets us ask which implementation details can remain hidden and which must remain meaningful at the interaction boundary. The notation uses established state-machine refinement and contract reasoning~\cite{abadi1991refinement,meyer1992contract}; the proposed contribution is their application to the relationship between shared application state and revisable delegated tasks.

\subsection{AIS as an Open Application System}
\label{sec:ais-model}

\begin{definition}[AIS execution model]
An \aisname realization is represented by an open labeled transition system
\begin{equation}
\mathcal{M}=\langle X,X_0,\Sigma,\longrightarrow\rangle,
\qquad
\Sigma=\Sigma_G\uplus\Sigma_I\uplus\Sigma_A\uplus\Sigma_E.
\label{eq:ais-model}
\end{equation}
Here $X$ is the global state space, $X_0$ contains initial states, and $x\xrightarrow{a}x'$ is a possible transition. The label classes distinguish direct GUI operations ($G$), intent-level commands and feedback ($I$), agent/runtime steps ($A$), and environmental events ($E$). Labels carry the relevant principal, task, and operation identifiers. The realization retains a conventional core and direct interaction path, integrates a goal-directed agent, and supports an \toolname realization over their shared application objects.
\end{definition}

For analysis, write a state as $x=(s,t,c,j,w)$. The conventional core state $s$ includes domain objects, versions, and applicable policies. The task store $t$ records goals, revisions, contextual bindings, delegated scope, and unresolved decisions. The control state $c$ identifies who may coordinate each task. The journal $j$ associates admitted effects with their payloads and observed outcomes. Admission is the host's acceptance of a specific effect for execution; it does not by itself establish that the effect has occurred. The remaining state $w$ represents agent internals, pending messages, and relevant environmental state. This is a logical decomposition; it neither requires five services nor assumes that the agent observes the entire state. In particular, an external effect can have occurred while its outcome remains unknown in $j$.

The open-system view matters because task execution does not suspend the application. A GUI edit, another user's action, a policy revocation, or a delayed provider response can interleave with an agent step. These events are part of the behavior to be reasoned about, rather than exceptions to a sequential dialogue. Labels are tagged by their role in the model: a user-facing agent report belongs to $\Sigma_I$, whereas internal planning belongs to $\Sigma_A$. The tag does not establish trustworthiness. An agent-originated operation still requires authorization, and a direct operation still obeys the application's business rules.

Architectural membership requires the coexistence of these responsibilities, not satisfaction of a safety property. A defective integration remains an \aisname instance. Similarly, neither local model execution nor a particular chat interface is a membership condition. These distinctions keep the model applicable to incremental adoption and distributed deployments without making quality claims true by definition.

\subsection{IIA as a Task-Level Abstraction}
\label{sec:iia-model}

The \toolname abstracts how delegated work develops over time. A task record can be written as $t_\kappa=(\kappa,r,g,b,d,U)$: task identity $\kappa$, revision $r$, declared goal $g$, contextual bindings $b$, delegated scope $d$, and unresolved decisions $U$. A binding refers to a stable object and the relevant version or view, rather than only its current screen position. The goal supplies task-specific acceptance conditions, which may remain partial until decisions in $U$ are resolved. This representation makes explicit what has been agreed; it does not assume that unrestricted natural-language meaning can be converted into a complete logical predicate.

\begin{definition}[Intent-Level Interaction Abstraction]
An \toolname specification is an abstract transition system $\mathcal{I}=\langle Y,Y_0,\Lambda,\Longrightarrow,V\rangle$. Here $Y$ contains task-level states, $Y_0$ their initial states, $\Lambda$ the event labels, and $\Longrightarrow$ the permitted abstract transitions. States summarize tasks, decisions, control status, and effect knowledge; events express goal submission, revision, endorsement, intervention, admission, and outcome reporting. $V_u(y)$ is the view permitted for principal $u$. A candidate realization supplies a state abstraction $\pi:X\rightarrow Y$ and an event abstraction $\alpha:\Sigma\rightarrow\Lambda\cup\{\varepsilon\}$, where $\varepsilon$ denotes an unobservable step.
\end{definition}

The central requirement is that a concrete step have a valid task-level interpretation:
\begin{equation}
\pi(X_0)\subseteq Y_0,
\qquad
x\xrightarrow{a}x'
\ \Longrightarrow\
\pi(x)\overset{\alpha(a)}{\Longrightarrow}\pi(x').
\label{eq:abstraction}
\end{equation}
For $\alpha(a)=\varepsilon$, the right-hand side means $\pi(x)=\pi(x')$. Internal reasoning can therefore remain hidden when it changes no task-relevant fact. A changed recipient, an admitted disclosure, or an acknowledged cancellation cannot be hidden in this way when it changes the contract's task state. A GUI event affecting a delegated task must have an abstract interpretation even though it did not originate in the \toolname interface. Where a useful abstraction needs history, the analysis state may be augmented with records of earlier endorsements and effects~\cite{abadi1991refinement}.

The abstraction describes a specification relationship, not an assertion that a runtime can automatically recover $\pi$. Implementations must supply the object, task, and outcome information on which it depends. Furthermore, $y$ represents shared task semantics while $V_u(y)$ deliberately restricts what each role can inspect. An affected recipient need not see an organizer's private context. User-facing feedback is faithful when its factual claims are justified by the corresponding authorized view; completeness of explanation and usability require additional judgment.

Control requires more than observation. A request to stop and an acknowledgement that stopping took effect are distinct abstract events. For a task $\kappa$, the controller epoch $k$ identifies a generation of control authority. An acknowledged stop invalidates that generation for subsequent admissions under $(\kappa,k)$, while permitting reconciliation of earlier admissions. An interface that displays ``stopped'' while the old controller can continue admitting work violates the abstraction's control semantics. Eventual acknowledgement is a separate progress requirement, conditional on service availability and communication assumptions.

\figref{fig:ais-architecture} brings the two levels together: the conventional core and agent constitute the execution system, while the \toolname exposes its task-level meaning through the abstraction $\pi$. The lower band uses the state spaces $X$ and $Y$ defined above.

\begin{figure}[!htbp]
\centering
\includegraphics[width=\linewidth]{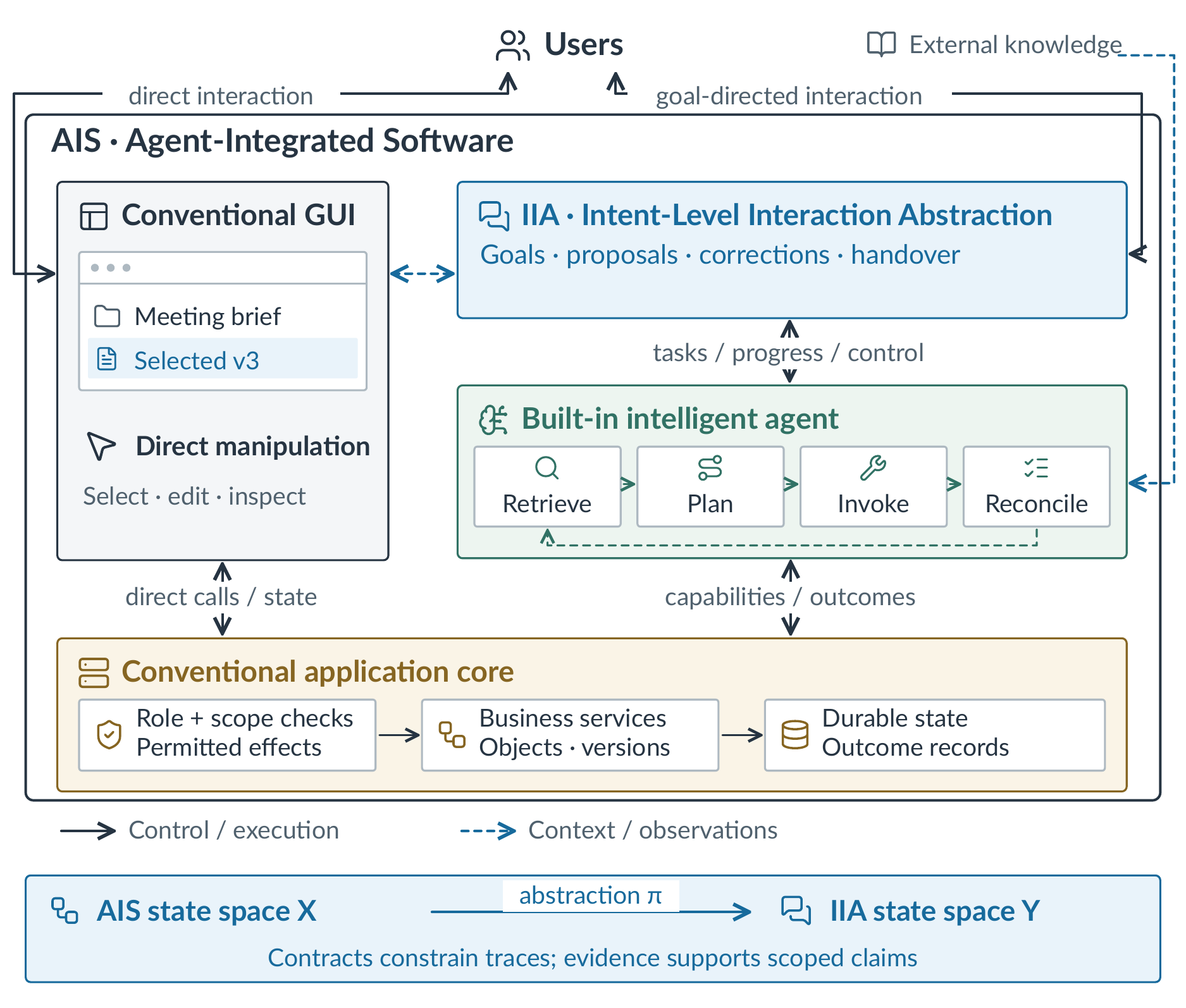}
\caption{AIS with continuing direct and intent-level interaction over a shared core. The IIA realization presents and controls task semantics; the built-in agent plans and invokes capabilities. The lower band relates the execution state space $X$ to the task state space $Y$. Control and effect records connect both levels to interaction contracts and assurance.}
\label{fig:ais-architecture}
\Description{A boundary labeled AIS contains a conventional GUI on the left, a layer realizing the Intent-Level Interaction Abstraction (IIA) above a built-in agent on the right, and a shared application core below. Users, potentially occupying different requester, owner, approver, and affected-party roles, exchange input and feedback through the two paths. Dashed arrows connect GUI context to the IIA layer and external knowledge to the agent. The agent cycles through retrieval, planning, invocation, and outcome reconciliation. The core contains permission checks, business services and versioned objects, and durable state with outcome records. A lower semantic band maps the AIS state space X to the IIA state space Y through an abstraction, with contracts constraining traces and evidence supporting claims.}
\end{figure}

\subsection{Interaction Contracts and Behavioral Conformance}
\label{sec:contract-model}

The same abstract event can have different obligations across applications. A spreadsheet commit changes shared cells; a delivery commit may disclose information outside the application. We capture these choices in an interaction contract
\begin{equation}
\mathcal{K}=\langle\mathsf{Pre},\mathsf{Step},\mathsf{Inv},
\mathsf{Post},\mathsf{Dep}\rangle.
\label{eq:contract-model}
\end{equation}
$\mathsf{Pre}$ gives operation preconditions, including authority; $\mathsf{Step}$ constrains task and control transitions; $\mathsf{Inv}$ states consistency obligations; $\mathsf{Post}$ defines what effect or report counts as fulfilling an operation; and $\mathsf{Dep}$ identifies the assumptions and versions on which these clauses depend. An implementation may realize these clauses through APIs, runtime guards, ordinary application code, and interaction design. The contract is an application-level semantic specification, independent of that choice.

For an effect $e$, its interaction--effect obligation is the relation
\begin{equation}
\Psi_e=\langle\kappa,r,\beta_e,\gamma_e,\ell_e,\epsilon_e\rangle,
\label{eq:obligation}
\end{equation}
linking the task and revision to reviewed payload/object bindings $\beta_e$, role-specific authority $\gamma_e$, controller identity and epoch $\ell_e$, and available outcome evidence $\epsilon_e$. The relation is temporal. Authority is checked at admission; later evidence determines which outcome can be reported. Revising a task can invalidate future admissions without erasing the attribution of an earlier committed effect. Consequently, consistency cannot be reduced to equality between the latest task state and every historical record.

Let $\mathsf{Tr}$ denote finite execution traces and let $\Pi$ be the trace projection induced by $\pi$ and $\alpha$, with unobservable repetitions removed. For stated environmental assumptions $H$, the proposed conformance obligation is
\begin{equation}
\Pi\bigl(\mathsf{Tr}(\mathcal{M}\mid H)\bigr)
\subseteq\mathsf{Tr}(\mathcal{I}\mid\mathcal{K}).
\label{eq:conformance}
\end{equation}
The right-hand side contains abstract traces permitted by the contract. This finite-trace condition addresses valid transitions and claims about observed effects. It does not establish eventual task completion: a runtime that remains silent can satisfy some safety clauses while failing a progress requirement. Such requirements need explicit fairness, availability, and deadline assumptions. Nor does conformance establish that the declared goal adequately represents the user's needs.

\begin{proposition}[Local validity does not establish interaction conformance]
Valid component calls and authenticated permissions are insufficient to establish Equation~\eqref{eq:conformance} under a contract requiring current endorsement bindings when task-relevant state can change before effect admission.
\end{proposition}
\begin{proof}
Take a proposal for recipients $\{a,b\}$ and an authentic endorsement of that proposal. Before admission, a GUI revision removes $b$. A send operation for $b$ can still satisfy its API schema and the caller's ordinary resource permission. Its projected trace nevertheless contains an admission inconsistent with the current task binding. Thus component validity permits a trace excluded by the interaction contract.
\end{proof}

The proposition locates the contribution of the framework. The missing relation crosses the boundaries of interface state, task interpretation, authority, and effects. Access control and transaction mechanisms remain essential, but their successful composition must be established with respect to that relation. A different mechanism that already preserves it is an alternative realization of the obligation.

\subsection{Foundations and Derivation of the Engineering Agenda}
\label{sec:evolution}
\label{sec:derivation}

The development summarized in \figref{fig:ais-evolution} explains why this relation is now exposed more broadly. Direct manipulation and mixed initiative preserve user participation~\cite{shneiderman1983direct,horvitz1999mixed}; demonstration and task shortcuts connect requests to existing application operations~\cite{li2017sugilite,arsan2021savant}; model-based planning and tool use support less predetermined procedures~\cite{yao2023react,schick2023toolformer,li2023apibank,wen2024autodroid,zhang2025ufo2}. Studies of product copilots and software built with foundation models (FMware) identify the resulting integration and lifecycle demands~\cite{parnin2025copilot,rajbahadur2026fmware}. The framework concentrates these developments into one question: which task-level relationships must survive changes in the underlying execution?

\begin{figure}[!htbp]
\centering
\includegraphics[width=\linewidth]{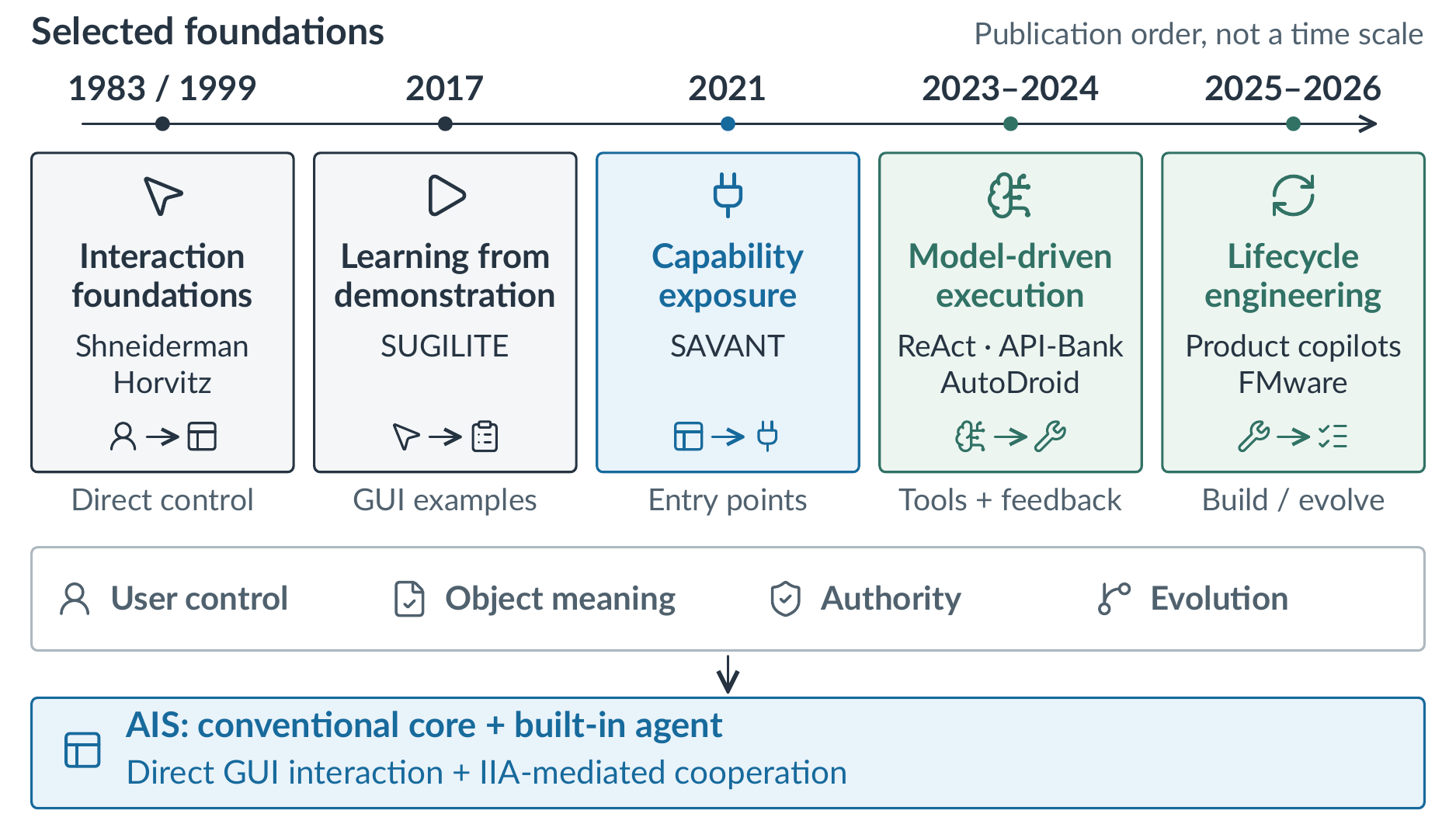}
\caption{Selected foundations of AIS: interaction~\cite{shneiderman1983direct,horvitz1999mixed}, demonstration~\cite{li2017sugilite}, capability exposure~\cite{arsan2021savant}, model-driven execution~\cite{yao2023react,li2023apibank,wen2024autodroid}, and lifecycle engineering~\cite{parnin2025copilot,rajbahadur2026fmware}. Publication groups indicate overlapping foundations, not replacement or a time scale.}
\label{fig:ais-evolution}
\Description{Five groups of publications run from direct manipulation in 1983 and mixed-initiative interaction in 1999 through demonstration in 2017, capability exposure in 2021, model-driven execution in 2023 and 2024, and lifecycle engineering in 2025 and 2026. Each group has a mechanism icon and representative publications. User control, object meaning, authority, and evolution remain common concerns. A final band defines AIS through a conventional core, built-in agent, and direct and IIA-mediated interaction.}
\end{figure}

The challenge areas follow from perturbing the terms in this framework. Changes to $g$, $b$, and $d$ expose task interpretation, interface, and context obligations; changes to $c$, authority in $s$, and effects in $j$ expose continuity, permission, and recovery obligations. Uncertainty about their observation or evolution exposes traceability and maintenance obligations. The next section develops these implications, followed by design principles and a research agenda. Their connections are synthesized at the end of the agenda. This organization is a retrospective design rationale, not an empirical taxonomy or an exhaustive derivation. Overlap is expected because a single object change can affect several relations.

\section{Engineering Implications of the Framework}
\label{sec:challenges}

The semantic framework shifts the engineering question from whether an agent can invoke an operation to whether its evolving execution remains a valid realization of the task. This section develops eight engineering implications. Each concerns information or control that crosses a component boundary; their separation identifies where development support and assurance are needed.

\subsection{C1: Task Abstraction and Degrees of Autonomy}
\label{sec:autonomy}

A task becomes actionable through the relationship between its declared goal $g$, delegated scope $d$, and unresolved decisions $U$. ``Prepare the materials'' can authorize retrieval and proposal construction while leaving disclosure undecided. Substituting a public version for a private attachment may satisfy a release policy but alter the user's intended result. The \toolname must preserve that distinction throughout the task, rather than allowing an inferred plan to stand in for an endorsement.

Autonomy is consequently a property of permitted transitions for a task and role, not a fixed level assigned to an agent. Participation in planning can improve the opportunity for correction~\cite{he2025plan}, but a review is consequential only if subsequent admissions remain bound to it. Requiring approval for every internal step can increase effort without narrowing the relevant effects. The design problem is to expose decisions that change the task's acceptance conditions or authority while allowing factual retrieval and reversible preparation within the established scope. This connects formal control semantics to the supervision question developed in \secref{sec:discussion-control}.

\subsection{C2: Bidirectional Interfaces}
\label{sec:interfaces}

The two abstraction levels require complementary interfaces. The application supplies capabilities through which the agent observes and changes domain state; the agent runtime supplies task-control services through which the application realizes the \toolname. Their contracts meet at task identity, contextual bindings, and effect evidence. \tabref{tab:interfaces} summarizes the division. An API can be callable without exposing enough information to relate its effects to a reviewed task, and an event stream can be well formed without establishing whether an operation committed.

\begin{table}[!htbp]
\caption{Complementary interfaces and the semantics needed to relate task interaction to application behavior.}
\label{tab:interfaces}
\centering
\small
\begin{tabular}{@{}>{\raggedright\arraybackslash}p{0.23\linewidth}>{\raggedright\arraybackslash}p{0.36\linewidth}>{\raggedright\arraybackslash}p{0.35\linewidth}@{}}
\toprule
Interface provider & Services available to the other side & Semantics to make explicit \\
\midrule
Application core to agent & Discover capabilities; read context and objects; subscribe to changes; validate authority; execute and verify operations. & Object identity and version, scope, prerequisites, side effects, errors, and retry behavior. \\
Agent runtime via the \toolname layer & Accept goals and context updates; expose progress and proposals; request input; return outcomes; suspend, cancel, or resume tasks. & Task identity and revision, pending user decisions, confirmed effects, and cancellation limits. \\
\bottomrule
\end{tabular}
\end{table}

API selection and structured invocation address the first part of this problem~\cite{qin2024toolllm,xie2025droidcall}. Semantic compatibility additionally requires stable meanings for preconditions, effects, cancellation, and outcomes. A schema-preserving change to a tool can therefore break the conformance relation. Conversely, an implementation of the Agent--User Interaction (AG-UI) protocol can transport state, edited approvals, interrupts, and resume events~\cite{agui2026protocol}, while the host supplies the admission checks and authoritative outcomes. The framework makes this division explicit so that protocol conformance is not mistaken for application-level conformance.

\subsection{C3: Context, Knowledge, and Memory}
\label{sec:context}

Context connects a task to an application state that may outlive the utterance that referred to it. The binding $b$ must distinguish object identity from position, a current value from an observed version, and an explicit instruction from an inferred preference. A screenshot can locate a document yet omit the identity needed to detect that the document was replaced. Cross-platform replay work by Yu \etal~\cite{yu2021lirat} illustrates the difficulty of preserving targets as interfaces change; preserving the intended business effect requires the additional task relation.

Personalization and retrieval extend these dependencies beyond the current screen~\cite{cai2025personalwab,shi2026tauknowledge}. A remembered preference can conflict with the current request, and an updated source can invalidate a derived summary. The challenge is to retain enough provenance for selective refresh and correction without turning memory into unrestricted retention. Read authority also differs from disclosure authority. The \toolname view $V_u$ and any transferred task summary must respect purpose and recipient scope even when the agent was permitted to read the underlying source.

\subsection{C4: Interaction Continuity and Shared State}
\label{sec:continuity}

Continuity requires a stable task meaning across changes in interface, controller, and execution location. \figref{fig:ais-coordination} illustrates how a GUI edit changes the admissible continuation of a delegated task. The important event is the change in the reviewed binding, not whether the correction arrived through chat or direct manipulation. Treating the two paths as separate sessions would lose precisely the relationship that the \toolname abstracts.

\begin{figure}[!htbp]
\centering
\includegraphics[width=\linewidth]{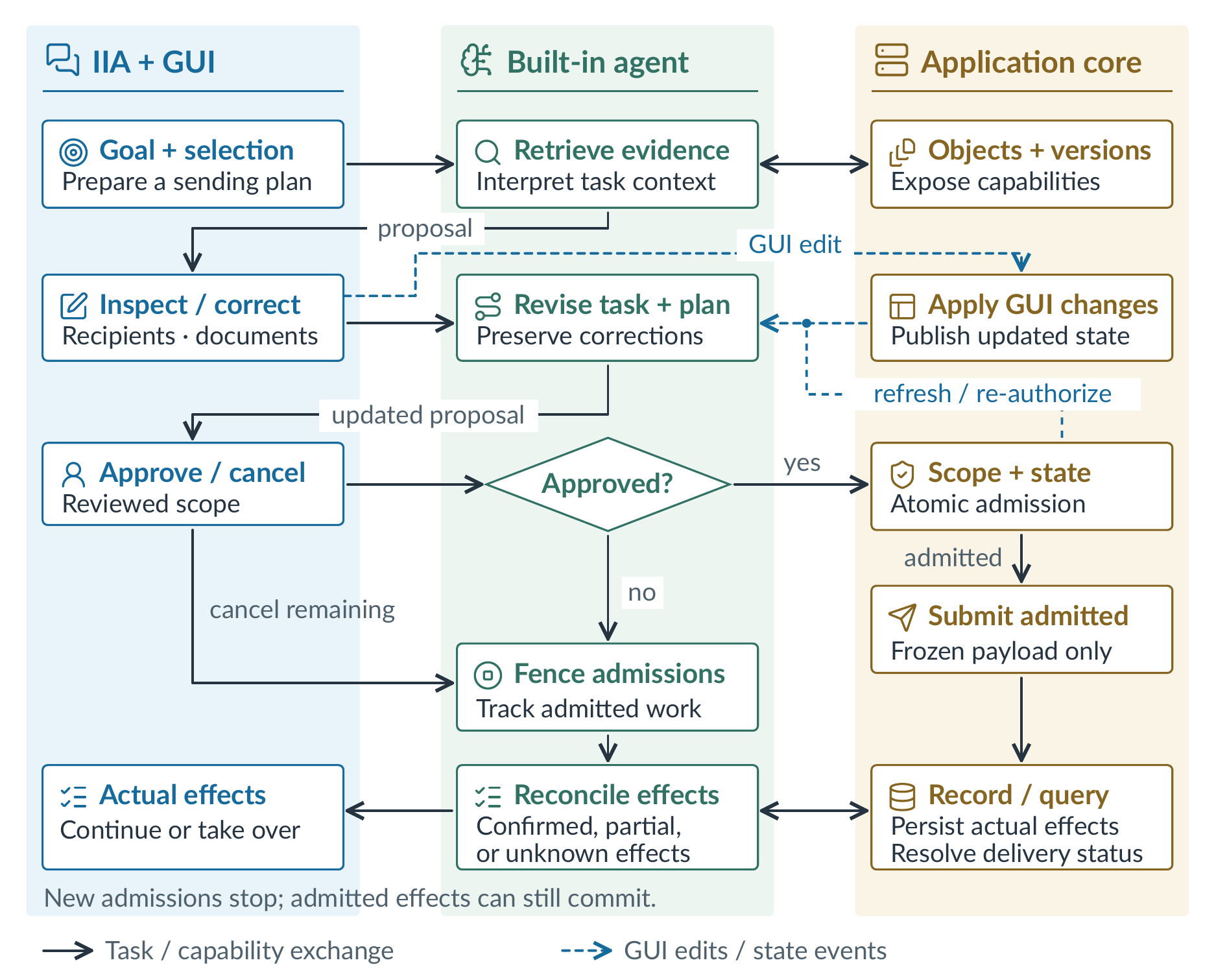}
\caption{Task revision and guarded execution across IIA/GUI interaction, agent coordination, and core operations. GUI edits revise the pending binding; changed conditions require refresh and, where necessary, renewed endorsement. An acknowledged stop fences new admissions while earlier admissions remain subject to outcome reconciliation. Each responsibility lane can involve multiple principals.}
\label{fig:ais-coordination}
\Description{The user selects materials and requests a sending plan. The agent retrieves evidence and presents a proposal. User corrections and GUI edits feed task revision and updated application state. The user then approves or cancels. Approved calls go to core scope and state checks before execution; changed conditions return to task revision. Non-approval and acknowledged cancellation lead to fencing new admissions and tracking admitted work. The core records or queries actual effects, the agent reconciles confirmed, partial, or unknown outcomes, and the IIA presents results and handover options. Cancellation does not automatically undo committed effects.}
\end{figure}

The term \emph{handover} covers two distinct changes: \emph{operational takeover by the same user} (H1) and \emph{transfer of responsibility to another person} (H2). \emph{Remote continuation or relocation} (H3) and \emph{recovery after failure} (H4) may accompany either, but do not themselves imply a change of responsible person. \tabref{tab:handover} compares these four cases. A controller lease records which runtime holds task-control authority and the conditions under which that authority expires. The four cases share a continuity record containing task revision, unresolved decisions, relevant object versions, grants, controller lease/epoch, and admitted effects with known outcomes. The record supports continuity only when the host can establish which controller may next admit work.

\begin{table}[!htbp]
\caption{Four continuity cases sharing a task record but requiring different authority and recovery semantics.}
\label{tab:handover}
\centering\small
\begin{tabular}{@{}>{\raggedright\arraybackslash}p{0.24\linewidth}>{\raggedright\arraybackslash}p{0.31\linewidth}>{\raggedright\arraybackslash}p{0.39\linewidth}@{}}
\toprule
Case & What changes & Required special handling \\
\midrule
H1: user takes over & Operational control moves from agent to the same user's direct interaction. & Fence agent admissions; show admitted/unknown effects and a continuation point. User edits invalidate affected plans, not the whole application. \\
H2: task transfers to another person & The accountable controller or task requester changes. & Revalidate the recipient's own rights and required endorsements; minimize transferred context; preserve attribution. Credentials and approvals are not transferred implicitly. \\
H3: remote continuation or relocation & The same task continues after UI disconnection, or moves between runtimes. & Persist task/effect identities. Reconnection reads authoritative state. Relocation establishes a new executor lease; client loss alone does not imply cancellation. \\
H4: recovery after failure & A runtime restarts with potentially incomplete observations. & Fence stale workers, rebuild from the journal, query unknown effects before resubmission, and disclose irrecoverable evidence gaps. \\
\bottomrule
\end{tabular}
\end{table}

A transfer therefore combines a semantic obligation with an enforcement obligation. The successor needs a rights-filtered account of what remains to be done, while the old controller must lose the ability to admit further effects under its former authority. Lost acknowledgements leave ownership unresolved; late provider responses leave effect knowledge incomplete. Neither uncertainty can be removed by copying a conversation. A closed client window is similarly compatible with an active remote task. These cases require different continuations even though they share a task identifier.

\subsection{C5: Permission, Safety, Security, and Privacy}
\label{sec:security}

Delegation crosses a relation among principals, rather than a single user--agent permission boundary. The initiator requests an outcome; resource owners govern the objects; approvers endorse particular effects; affected parties receive their consequences; and an executor/controller coordinates the work. A person may occupy several roles, but their authorities do not merge automatically. In the document scenario, an organizer's request does not supersede a team's release policy or a disclosure officer's approval.

For an effect $e$ in state $x$, the authority condition is conjunctive:
\begin{equation}
\begin{split}
\operatorname{permit}(e,x)={}&\operatorname{taskGrant}(e,x)
\land\operatorname{resourcePolicy}(e,x)\\
&{}\land\operatorname{requiredEndorsements}(e,x)
\land\operatorname{executorScope}(e,x).
\end{split}
\label{eq:roles}
\end{equation}
Attribute-based access control supplies established mechanisms for subject, resource, action, and environment conditions~\cite{hu2014abac}. The framework additionally relates these conditions to the current task binding. Owner consent may be represented by a standing grant; affected-party consent and separation of duties become required predicates where the domain policy demands them. Changes to ownership, recipients, account, or controller require revalidation of the dependencies concerned.

An endorsement must also have an origin distinct from agent inference. Execution delegation cannot include the ability to write the approval ledger or activate an indistinguishable approval input. Otherwise, an authenticated session establishes an account identity but not a person's endorsement of the particular effect. The distinction between data and authority is equally important for retrieved material. Prompt-injection defenses and runtime constraints provide relevant enforcement techniques~\cite{debenedetti2024agentdojo,debenedetti2025camel,wang2026agentspec}; their scope depends on complete mediation across reachable execution paths~\cite{saltzer1975protection}. An unrestricted alternative shell or database connection defeats a claim covering only guarded business APIs.

Safety, security, and privacy impose related but different obligations. A permitted action can still be mistaken; an adversarial action can redirect control; an otherwise legitimate result can disclose excessive information. Role-specific views, logs, and derived memory therefore need their own disclosure and retention conditions. Preserving attribution during a handover does not imply transferring the previous controller's credentials or private context.

\subsection{C6: Execution Semantics and Failure Recovery}
\label{sec:recovery}

The distinction between application state $s$ and effect knowledge in $j$ is essential under failure. A timeout can leave the runtime unable to distinguish non-execution from an already committed effect. Treating both as failure permits duplicate execution; treating both as success produces misleading feedback. Stable effect identifiers and authoritative status queries make that uncertainty manageable, while the contract determines whether replay is safe.

Cancellation changes future admissibility at a particular boundary. It does not reverse an effect already admitted to an external provider. Compensation introduces a new domain operation with its own authority and failure conditions, as in sagas~\cite{garciamolina1987sagas}. A calendar invitation can be withdrawn, but its notification may already have been read. A spreadsheet undo can restore prior values yet incorrectly overwrite a collaborator's later edit. Recovery must therefore preserve the relation among original effects, subsequent changes, and new compensating actions; a generated inverse command is insufficient.

\subsection{C7: Observability, Debugging, and Responsibility}
\label{sec:observability}

The abstraction provides a natural basis for diagnosis: a failure occurs where a concrete trace no longer has the task-level interpretation promised by the contract. Agent debugging and lifecycle observability offer useful inspection mechanisms~\cite{epperson2025debugging,dong2024agentops}. For \aisname, traces must additionally connect task revisions and endorsements to host admissions and outcome records. Otherwise, an explanation can describe why the agent acted without establishing whether that action was authorized or what it changed.

Different readers need different projections of this evidence. Users need the effects that occurred and the decisions that remain. Developers need the dependency or transition responsible for a mismatch. The evidence includes approvals, admission witnesses, receipts, \etc A model-generated explanation can help interpret those records but cannot replace them. Replay also needs controlled versions and external state before it can support causal diagnosis. This separation makes evidence useful for accountability while leaving institutional responsibility to the relevant organizational arrangements.

\subsection{C8: Architecture, Joint Evolution, and Operating Costs}
\label{sec:evolution-challenge}

The correspondence between $\mathcal{M}$ and $\mathcal{I}$ can change even when visible API types do not. A model replacement changes possible plans; a host update changes object meaning; an adapter revision changes admission or reporting behavior. Experience with learned-system engineering and technical debt motivates explicit management of such dependencies~\cite{amershi2019seml,sculley2015debt,rajbahadur2026fmware}. The framework makes their consequence precise: some previously supported traces or assurance claims may no longer describe the revised implementation.

This raises an architectural choice for existing applications. Stable references, guarded writes, task-control hooks, and outcome records have development and maintenance costs. Integrating an agent without those services may still offer useful preparation or retrieval, but supports weaker effect guarantees. Deployment further changes latency, information exposure, and continuity assumptions. A viable integration must therefore be judged against total development, computation, supervision, and recovery effort, including the conventional workflow it complements.

\section{Interaction Contracts and Continuous Assurance}
\label{sec:quality}

The framework makes a distinction that conventional task-success measures can obscure. An execution may reach the requested final state through an unauthorized intermediate effect, or preserve every authorization condition while failing to achieve a useful outcome. Dependability therefore requires both a valid interaction history and an adequate task result. Interaction contracts state selected obligations over that history; assurance identifies the evidence for believing they hold in a particular realization.

\subsection{Quality Across the Abstraction Boundary}
\label{sec:quality-model}

At the \toolname level, quality concerns the fidelity and usability of task expression, proposals, interventions, and feedback. At the \aisname level, it additionally concerns planning, domain effects, authority, availability, and evolution. Existing software and AI-system quality models supply broader terminology~\cite{iso25010,iso25059}; \tabref{tab:quality} specializes the concerns to the abstraction boundary. A correct update with misleading feedback and a clear preview followed by an incorrect update are distinct failures of the same relationship.

\begin{table}[!htbp]
\caption{Quality concerns across the IIA abstraction and its AIS realization. Primary loci identify the information or mechanisms needed to assess each concern.}
\label{tab:quality}
\centering
\small
\begin{tabular}{@{}>{\raggedright\arraybackslash}p{0.28\linewidth}>{\raggedright\arraybackslash}p{0.24\linewidth}>{\raggedright\arraybackslash}p{0.42\linewidth}@{}}
\toprule
Quality concern & Primary locus & Example of a failure \\
\midrule
Intent and user control & \toolname layer and agent & Ignore a recipient corrected through the GUI. \\
Grounded planning & Agent and context services & Use an obsolete conversation as the basis for a plan. \\
Business-effect correctness & Core and agent runtime & Repeat a write after an ambiguous timeout. \\
Feedback fidelity & \toolname layer and core evidence & Report delivery before the core confirms it. \\
Consistency and reliability & \aisname across both paths & Resume from task state that disagrees with application state. \\
Safety, security, and privacy & \aisname across trust boundaries & Treat retrieved text as authority to disclose data. \\
Availability and efficiency & \aisname deployment & Lose the direct-operation fallback when the model is unavailable. \\
Maintainability and observability & \aisname lifecycle & Silently lose a capability after an interface update. \\
\bottomrule
\end{tabular}
\end{table}

Task acceptance must also depend on the authorized stage. A preparation request can succeed with an inspectable proposal; an execution request can appropriately lead to clarification or justified non-execution. Such outcomes cannot be evaluated by counting completed writes. Conversely, conformance to permission and control clauses does not establish that a summary is accurate or a refactoring useful. The framework separates these judgments so that strong evidence for one property cannot silently stand in for another.

The placement of responsibility follows the location of the required knowledge. Object identity and commit evidence belong at the host boundary; task revision and endorsement must remain connected to them; user-facing explanations must draw on the resulting records. This follows the spirit of end-to-end reasoning about application-specific correctness~\cite{saltzer1984endtoend}. It also explains why a single aggregate quality score is inadequate: improved response fluency cannot compensate for an unauthorized disclosure, and fewer interruptions do not necessarily improve user control.

\subsection{A Compact Contract for Reviewed Disclosure}
\label{sec:contract}
\label{sec:conventions}

\contractref{ct:share} instantiates $\mathcal{K}$ for the meeting-materials scenario. Its compact form separates the scenario's binding and authority from reusable admission, control, and reporting clauses. The values are symbolic: task $T42$ at revision 7, controller epoch 3, a public document view at version 12, and recipients $a$ and $b$. $H7$ identifies the reviewed view and $D7$ its canonical proposal digest. The contract gives semantic obligations rather than prescribing a configuration language or a runtime architecture.

The contract's invariants are expressed by four clauses. \textbf{K1, endorsement binding}, relates an admitted payload to the current task revision, reviewed objects/views, recipients, and effect scope. \textbf{K2, current authority}, requires the conjunction in Equation~\eqref{eq:roles} at admission. \textbf{K3, controlled admission}, requires running task state, an authenticated controller lease, and the active epoch, with the acceptance decision recorded atomically. \textbf{K4, outcome fidelity}, permits a committed report only when authoritative evidence establishes the specified effect. These clauses specialize $\mathsf{Pre}$, $\mathsf{Inv}$, and $\mathsf{Post}$; the revision and control rules specialize $\mathsf{Step}$.

\begin{interactioncontract}[!htbp]
\caption{Reviewed document disclosure, \texttt{share-materials/v1}.}
\label{ct:share}
\centering\small
\begin{tabular}{@{}>{\raggedright\arraybackslash}p{0.19\linewidth}>{\raggedright\arraybackslash}p{0.76\linewidth}@{}}
\toprule
Binding & $\kappa=T42$, $r=7$, $k=3$; document 7, version 12, public view $H7$; recipients $\{a,b\}$; frozen proposal $D7$. Only reviewed materials and recipients may be delivered. \\
Authority & Organizer initiates; team owns the document; event chair and disclosure officer approve; recipients are affected; agent service acts for the organizer under lease $L3$. Owner grant $G18$ and policy \texttt{policy-v4} apply. \\
Endorsement & A host ledger records endorsements through a non-delegated channel, bound to $(\kappa,r,D7,G18,\texttt{policy-v4})$ and required roles. Recheck scope, revocation, and expiry; task endorsement expires no later than ten minutes after recording. \\
Preparation & Read authorized context and public views; preview without delivery. Material revision increments $r$ and invalidates prior task endorsements. Complete, valid approval enables running; unspecified operations or transitions are rejected. \\
Admission & Atomically check K1--K3 and journal a stable effect ID for each recipient. Submit only its frozen payload, using the same provider idempotency key when safe to retry. Status queries obey read permissions. \\
Control & Cancel atomically stops the task, revokes task endorsement, and advances $k$. H1/H2 fence the old lease and require revalidation and acceptance; H3/H4 revalidate delegation and reconcile before resuming affected writes. Preserve unresolved effect IDs. \\
Outcome & K4 governs every report. Distinguish not admitted, admitted, submitted, committed, failed, and unknown. Commitment means provider-recorded delivery to the named inbox; submission acceptance alone is insufficient. \\
Dependencies & Contract, host, adapter, policy, role grants, and controller state. An unknown submission is queried before replay; cancellation fences new admissions while earlier admitted effects may still commit. \\
\bottomrule
\end{tabular}
\end{interactioncontract}

Standing owner grants remain usable after revalidation; revising a task does not itself revoke them.

Let $B_e$ be the reviewed binding for effect $e$, $A_e$ its required endorsements, and $\ell_e$ its authenticated controller lease with epoch $k_e$. For state $x$ immediately before admission,
\begin{equation}
\begin{split}
\operatorname{admit}(e,x)\Longrightarrow{}&
\operatorname{current}(B_e,x)\land\operatorname{endorsed}(A_e,B_e,x)\\
&{}\land\operatorname{permit}(e,x)\land\operatorname{running}(\kappa,x)\\
&{}\land\operatorname{validLease}(\ell_e,x)
\land k_e=k_{\mathrm{active}}(\kappa,x).
\end{split}
\label{eq:admission}
\end{equation}
The gate evaluates this predicate and records the admitted payload within one host-controlled atomic boundary. Merely supplying the latest epoch number does not authenticate a controller. Likewise, retaining a conversational thread does not establish a current endorsement. A stable effect ID cannot be rebound to a different payload, and an unresolved submission cannot become a supposedly new action by receiving a fresh ID after a task revision.

\begin{proposition}[Planner-independent admission safety]
\label{prop:gate}
Suppose all admissions in the contract's scope pass through a trusted gate that atomically enforces Equation~\eqref{eq:admission}, the bindings and authority records are authoritative, and only this gate can append immutable admission records to an initially valid journal. Every admitted effect then satisfies K1--K3 at its admission point, independently of how the planner selected it.
\end{proposition}
\begin{proof}
The initial journal satisfies the property by assumption. A transition that appends no admission preserves it. A transition that appends one first establishes the admission predicate and binds the accepted payload to that record atomically. The property therefore holds for every finite sequence of admissions. Subsequent policy changes affect later decisions without changing what held at an earlier admission point.
\end{proof}

This limited result explains a useful architectural separation. The model can propose actions nondeterministically while selected admission invariants remain properties of the host boundary. The result establishes neither K4 nor goal suitability, progress, or absence of unmediated effects. It also exposes the assumptions that an integration must justify before claiming such a separation. A gate that reads stale policy state or an adapter that bypasses the gate does not satisfy the premises.

Task control and effect state remain independent. In this contract, commitment means delivery recorded by the provider, not that a recipient has read or understood the material. A stopped task can contain a committed delivery. If the provider exposes only submission acceptance, the promised postcondition must be narrowed or delivery knowledge remains unknown. Provider-level duplicate prevention similarly requires an idempotency guarantee or an equivalent queryable operation identity; the host journal alone cannot create exactly-once remote execution.

\subsection{Revision, Interruption, and the Meaning of a Valid Trace}
\label{sec:contract-trace}
\label{sec:principles}

The earlier counterexample becomes precise under the contract. Revision 7 proposes $\{a,b\}$; a GUI edit removes $b$ and produces revision 8. A delayed endorsement for revision 7 fails K1. Once the revised proposal receives the required approval, the gate may admit delivery $E_a$. A lost provider response leaves its outcome unknown under K4. Cancellation then advances the epoch and excludes new admissions, but a late receipt can still establish that $E_a$ committed. A faithful final account includes both the stopped remainder and the completed delivery.

Four design principles follow from the trace. \textbf{P1, preserve referential integrity}, keeps revised goals attached to the right objects and effects. \textbf{P2, separate interpretation from authorization}, makes endorsement depend on accountable roles. \textbf{P3, make control operational}, gives interventions consequences for admissible behavior. \textbf{P4, ground feedback in outcomes}, keeps the user's understanding aligned with evidence. In the compact contract these principles become K1--K4, rather than an additional independent checklist.

The same trace can guide development without pretending to be an experiment. It identifies the records an adapter must expose and the distinctions a test oracle must observe. Injecting a stale approval, timeout, or delayed receipt would exercise a different clause. More permissive contracts could reuse unaffected step-level endorsements after revision; doing so would require evidence that the dependency analysis preserved their scope. The conservative invalidation in \contractref{ct:share} makes that trade-off explicit.

\subsection{Assurance as a Versioned Argument}
\label{sec:assurance}

Continuous assurance concerns the standing of a claim about the realization, rather than the mere presence of a monitor. For each claim, define a record
\begin{equation}
\mathcal{A}=\langle q,\Omega,D,E,R,S\rangle,
\qquad
S\in\{\mathsf{supported},\mathsf{refuted},\mathsf{insufficient}\}.
\label{eq:assurance-record}
\end{equation}
Here $\Omega$ identifies the task, interval, or release scope; $D$ records assumptions and versioned dependencies; $E$ contains evidence with producer and observation time; and $R$ is the rule used to assess it. Support is conditional on those premises. A witnessed violation refutes a claim, whereas an unavailable receipt, expired premise, or coverage gap leaves insufficient support. The distinction prevents missing evidence from being interpreted as a successful check.

The assurance argument separates five claims: \emph{endorsement and authority at admission} (Q1), \emph{effective stopping of new admissions} (Q2), \emph{evidence-grounded outcome reporting} (Q3), \emph{continuity without silent replay} (Q4), and \emph{release behavior on a declared assessment scope} (Q5). \tabref{tab:assurance} states their scope, evidence, and checking rules.

\begin{table}[!htbp]
\caption{Five scoped assurance claims with evidence, checking rules, and conditions requiring reassessment.}
\label{tab:assurance}
\centering\small
\begin{tabular}{@{}>{\raggedright\arraybackslash}p{0.24\linewidth}>{\raggedright\arraybackslash}p{0.27\linewidth}>{\raggedright\arraybackslash}p{0.43\linewidth}@{}}
\toprule
Claim and scope & Evidence producer and artifact & Check, invalidation, and response \\
\midrule
Q1: every admitted effect in the observed interval matches endorsement and authority at admission. & Host admission journal; authenticated approval/role ledger; coverage of write gates. & Check K1/K2 and the admission witness. A mismatching record refutes Q1. Lost coverage removes support. Changed policy/ownership or stale grants require a new scope and revalidation for subsequent admissions. \\
Q2: after an acknowledged stop, the old controller admits no new effects. & Task store: stop acknowledgement, lease/epoch transition; admission journal. & Check runnable state, controller lease, and K3 ordering at the gate. A later stale admission refutes Q2. Missing acknowledgement/order evidence is insufficient; fence and reconcile. \\
Q3: every reported committed effect has authoritative matching outcome evidence. & Core commit record or provider receipt; effect journal; user-facing report. & Match effect ID and frozen payload under K4. A premature completion report refutes Q3. A timeout leaves the effect outcome unknown; query before reporting or resubmitting. \\
Q4: a control transition preserves unresolved work without silent replay. & Transfer/recovery record; old/new controller acknowledgements; stable effect IDs. & Check H1--H4 obligations. Missing transfer acknowledgement or an unreconciled submission prevents resumption of affected writes; changed recipient rights invalidate transferred context. \\
Q5: a release meets specified task/interaction properties on its declared regression scope. & Version-bound scenario tests, negative traces, and human-reviewed task outcomes. & Check stated acceptance rules and record coverage. A relevant model, adapter, host, or specification change invalidates affected evidence; rerun or narrow the claim. Passing samples do not establish universal correctness. \\
\bottomrule
\end{tabular}
\end{table}

The timeout trace can support Q1 while leaving delivery unknown. Reporting that uncertainty preserves Q3, whereas asserting delivery without evidence violates it. A late receipt does not refute Q2 when it concerns an admission preceding the stop. The claims therefore express separable obligations that a generic task-success indicator would collapse.

Changes determine which arguments must be reconsidered. Let $\operatorname{Dep}(q,\Omega)$ denote the dependencies needed to support claim $q$ over scope $\Omega$, and let $\Delta$ be the set of changed dependencies. A candidate invalidation rule is
\begin{equation}
\operatorname{Dep}(q,\Omega)\cap\Delta\ne\varnothing
\quad\Longrightarrow\quad
\operatorname{reassess}(q,\Omega'),
\label{eq:invalidation}
\end{equation}
where $\Omega'$ is the proposed scope after the change. An argument remains applicable only if its relevant premises remain valid or are re-established. Historical evidence retains its original attribution; it does not automatically support the new scope. Dependency analysis is itself an assurance obligation: an omitted dependency makes selective invalidation unsound.

The conditional separation in Proposition~\ref{prop:gate} suggests that not every change needs the same response. Replacing a planner may invalidate behavioral evidence for Q5 while preserving an independently established invariant of an unchanged gate. Adding an unmediated write capability invalidates that gate's coverage premise and affects Q1/Q2. Changing an outcome adapter can undermine Q3 even if planning remains unchanged. This is the practical meaning of continuous assurance: a maintained argument across relevant changes, rather than universal retesting after every token.

Design analysis, release evaluation, and operational records offer different evidence. Production-readiness and lifecycle-observability work provides useful foundations~\cite{breck2017mltest,dong2024agentops}; the framework binds their artifacts to the same interaction--effect relation. Its research challenge is to make that binding economical and trustworthy enough to guide maintenance. Evidence provenance, retention cost, monitor completeness, and acceptable uncertainty will influence whether the approach remains usable beyond a small example.

\subsection{What Quality Assessment Must Observe}
\label{sec:evaluation-role}

Assessment must observe task-relevant transitions as well as final outcomes. Executable web, desktop, and mobile environments provide foundations for outcome-based evaluation~\cite{zhou2024webarena,xie2024osworld,rawles2025androidworld}. Stateful tools, business tasks, and side-effect checks broaden that view~\cite{lu2025toolsandbox,yao2025taubench,huang2025crmarena,trivedi2024appworld}. For \aisname, the distinguishing need is to exercise the correspondence across direct and delegated interaction: revisions, role changes, interrupted admissions, and unresolved effects.

Scenario-guided GUI testing~\cite{yu2026scengen} and dual-control environments~\cite{barres2025tau2} suggest ways to construct such cases. Equivalent intentions need not produce identical click sequences, and a justified refusal need not be a failed task. The relevant oracle compares permissible business behavior under equivalent initial conditions. Real users remain necessary for evaluating whether the \toolname makes consequential differences inspectable; user simulation alone can misrepresent human interaction~\cite{seshadri2026simulation}. The next generation of assessment should connect behavioral validity with the effort and decision quality of the people supervising it.

\section{A Research Agenda from Development to Assurance}
\label{sec:directions}

The framework suggests a research program organized around constructing and maintaining the correspondence in Equation~\eqref{eq:conformance}. Application analysis supplies the concrete semantics; interaction design supplies the abstract task model; contracts relate them; evidence supports their continued use. The six directions below develop this program from application analysis to the maintenance of evidence. Their shared artifacts create an opportunity for development assistance and quality assessment to reinforce one another.

\subsection{R1: Recovering Application Abstractions and Interface Contracts}
\label{sec:future-interfaces}

R1 concerns the construction of $\pi$ and the two interface contracts from an existing application. The central difficulty is recovering semantic obligations rather than merely enumerating callable operations. A tool should distinguish an object's stable identity, the conditions that authorize its modification, and the evidence that establishes the resulting effect. GUI analysis, testing, and intent inference offer candidate sources for this information~\cite{yu2025vision,yu2024platform,yu2026intention}. Their outputs would need explicit links to host implementations and reviewable uncertainty.

This direction also concerns abstraction adequacy. A representation that omits a recipient or an intermediate disclosure may admit a false correspondence even when every represented transition is valid. Development support should therefore help identify distinctions that the domain considers consequential, expose unsupported clauses, and maintain the resulting schema as the application evolves. The relevant benefit is a reduction in total specification and maintenance effort while retaining those distinctions.

\subsection{R2: Maintaining Task Meaning under Context Change}
\label{sec:future-context}

R2 concerns how $g$, $b$, $d$, and $U$ evolve together. Rebuilding the entire task after every GUI edit is disruptive, while preserving every prior endorsement is unsafe. A useful task representation would identify which decisions depend on a changed source, object version, preference, or recipient. It could then distinguish a harmless presentation change from one that requires a new endorsement or invalidates an earlier plan.

The difficult cases involve dependencies absent from a formal schema. A renamed document can retain its meaning, whereas an unchanged filename can conceal different disclosure implications. Research is needed on combining explicit domain relations, inferred dependencies, and user correction without treating model confidence as proof of equivalence. Such representations must also support rights-filtered continuity in H1--H4, so that preserving task meaning does not imply preserving all private context.

\subsection{R3: Enforcing Delegation across Heterogeneous Effect Boundaries}
\label{sec:future-delegation}

R3 concerns the realization of $\mathsf{Pre}$, $\mathsf{Step}$, and $\mathsf{Inv}$ where operations cross APIs, GUI automation, and remote services. Runtime enforcement offers mechanisms for selected rules~\cite{wang2026agentspec}; the larger challenge is ensuring that they refer to the current task and the actual admission boundary. A capability description should reveal whether the provider supports atomic checks, duplicate prevention, status queries, and compensation.

This suggests conformance profiles with explicitly different guarantees. A local versioned edit can support stronger interference checks than an external delivery service with opaque commitment. Comparing these profiles through adversarial traces would clarify which guarantees survive retries, revocation, and controller transfer. The research objective is a principled basis for accepting, restricting, or declining delegated effects under the provider's actual semantics.

\subsection{R4: Assisting Development through Shared Semantic Artifacts}
\label{sec:future-development}

R4 concerns development tools that use the same contract artifacts as assurance. An admission predicate can guide adapter generation; a prohibited transition can generate a negative test; a failed assurance claim can locate a missing record or an invalid assumption. Existing work on test generation and migration, declarative language-model programs, and prompt comparison provides complementary techniques~\cite{yu2023scripts,khattab2024dspy,arawjo2024chainforge}. The distinctive opportunity is to connect these activities through maintained task and effect semantics.

Generated specifications cannot serve as independent oracles merely because they are expressed formally. A tool that infers a contract from faulty code may reproduce the fault. Research should distinguish recovered implementation behavior, intended requirements, and reviewed contract clauses, while measuring the effort of resolving disagreements. Developer experience should include subsequent host changes and evidence maintenance, not only the speed of producing an initial integration.

\subsection{R5: Assessing Conformance and Maintaining Evidence}
\label{sec:future-quality}

R5 concerns which evidence justifies a scoped conformance claim. Tests can explore traces, static or runtime analysis can establish selected invariants, and human assessment can judge task adequacy and inspectability. Their scopes differ. Metamorphic relations~\cite{chen2018metamorphic}, for example, could test whether meaning-preserving request variants retain permitted business effects, but identifying such variants requires domain judgment beyond the tested model.

Evidence dependencies introduce a second research problem. Selective reassessment is useful only if the dependency model catches relevant changes without overwhelming developers with unnecessary work. Longitudinal evaluations should therefore consider missed invalidations, unnecessary reassessments, diagnostic value, and maintenance cost. A system that produces frequent passing checks while silently losing the premises of its claims would fail the purpose of continuous assurance.

\subsection{R6: Preserving Human Authority and Sustainable Integration}
\label{sec:future-ecosystem}

R6 concerns the allocation of authority and burden across the people affected by delegated work. The requester, owner, approver, and operator may have different interests; a common protocol does not resolve them. Portable role and consent descriptions could help preserve obligations across organizational boundaries, provided that changes of context trigger explicit revalidation rather than implicit inheritance.

The same perspective applies to deployment and model substitution. Interchangeable interfaces can conceal different data exposure, latency, operating cost, and recovery behavior. Sustainable integration requires assessing those changes together with the effort of supervision and maintenance. The long-term objective is software in which delegation remains useful and accountable as both the host and its intelligent components evolve.

Taken together, the directions connect the implications developed in \secref{sec:challenges} with the principles derived in \secref{sec:principles}. \tabref{tab:derivation} records these connections alongside their literature premises and motivating perturbations. \figref{fig:ais-roadmap} projects the same challenge--direction links into a visual overview. The links are a design synthesis: they identify shared research work without assigning empirical weights or excluding additional relationships.

\begin{table}[!htbp]
\caption{Synthesis of the engineering implications (C1--C8, \secref{sec:challenges}), design principles (P1--P4, \secref{sec:principles}), and research directions (R1--R6, \secref{sec:directions}). The literature premises and perturbations explain the proposed connections.}
\label{tab:derivation}
\centering\small
\begin{tabular}{@{}>{\raggedright\arraybackslash}p{0.16\linewidth}>{\raggedright\arraybackslash}p{0.25\linewidth}>{\raggedright\arraybackslash}p{0.33\linewidth}>{\raggedright\arraybackslash}p{0.16\linewidth}@{}}
\toprule
Area & Literature premise & Perturbation and derived obligation & Principles; directions \\
\midrule
\textbf{C1} Task abstraction & Goal uncertainty and user participation~\cite{horvitz1999mixed,he2025plan}. & Change ``prepare'' to ``send'': distinguish inferred intent from delegated effects. & P1, P2; R1, R2, R3, R6 \\
\textbf{C2} Interfaces & Tool discovery/calling and product integration~\cite{li2023apibank,parnin2025copilot}. & Change a tool's effect without its schema: specify capability and task-control contracts. & P1--P4; R1, R4, R5 \\
\textbf{C3} Context and memory & Personalized and knowledge-dependent actions~\cite{cai2025personalwab,shi2026tauknowledge}. & Replace a document or preference: retain identity, provenance, and invalidation dependencies. & P1, P2; R2, R3, R5 \\
\textbf{C4} Continuity & Shared plans and stateful interaction~\cite{feng2026cocoa,lu2025toolsandbox}. & Edit in the GUI or change controller: reconcile task state and ownership. & P1, P3; R2, R3, R5 \\
\textbf{C5} Authority and safety & Untrusted tool content and runtime policies~\cite{debenedetti2024agentdojo,wang2026agentspec}. & Introduce an external recipient or revoke an owner grant: revalidate authority at admission. & P2; R3, R5, R6 \\
\textbf{C6} Execution and recovery & Outcome-based evaluation and compensation~\cite{trivedi2024appworld,garciamolina1987sagas}. & Lose a delivery response: distinguish unknown, failed, committed, and compensable effects. & P3, P4; R1, R3, R5 \\
\textbf{C7} Observability & Agent debugging and lifecycle traces~\cite{epperson2025debugging,dong2024agentops}. & Present a success message without a receipt: preserve evidence lineage for diagnosis and feedback. & P4; R4, R5, R6 \\
\textbf{C8} Architecture and evolution & Learned-system dependencies and production maintenance~\cite{sculley2015debt,rajbahadur2026fmware}. & Replace a model, adapter, or host version: invalidate affected behavioral evidence. & P1, P4; R1, R4, R5, R6 \\
\bottomrule
\end{tabular}
\end{table}

\begin{figure}[!htbp]
\centering
\includegraphics[width=\linewidth]{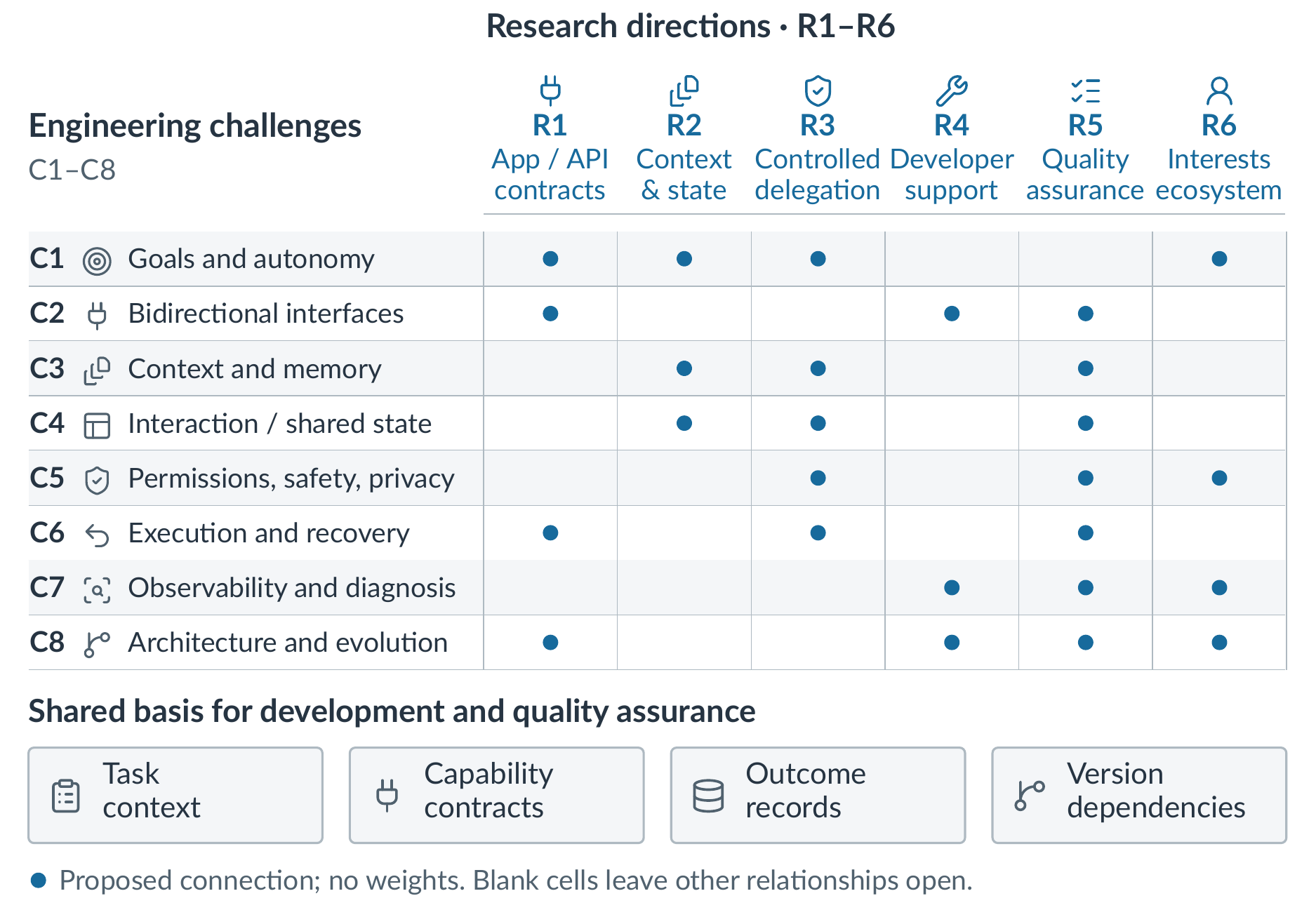}
\caption{Connections between engineering implications C1--C8 and research directions R1--R6, grounded in \tabref{tab:derivation}. The 26 unweighted links identify shared research concerns; labels abbreviate the subsection topics. Task context, contracts, outcome records, and version dependencies connect development assistance with assurance; blank cells do not exclude further relationships.}
\label{fig:ais-roadmap}
\Description{An eight-row, six-column matrix uses labeled icons and 26 dots to connect challenges with interfaces, context and state, controlled delegation, developer support, quality assurance, and user interests and ecosystem. The dots have no weights or empirical interpretation. Four tiles below identify shared task context, capability contracts, outcome records, and version dependencies.}
\end{figure}

\section{Discussion}
\label{sec:discussion}

\subsection{Scope across Software Domains}
\label{sec:discussion-costs}

The framework's scope follows from the relationships it represents, rather than the surface form of a task. \tabref{tab:domains} instantiates the same questions for four domains: what constitutes a binding, whose authority matters, which event changes the permissible continuation, and what evidence establishes an outcome. The examples are hypothetical design probes. Their value is to expose different interpretations of the shared model, including cases in which a host cannot support a desired contract.

\begin{table}[!htbp]
\caption{Domain interpretations of task bindings, authority, and continuity in four hypothetical design probes.}
\label{tab:domains}
\centering\small
\begin{tabular}{@{}>{\raggedright\arraybackslash}p{0.19\linewidth}>{\raggedright\arraybackslash}p{0.34\linewidth}>{\raggedright\arraybackslash}p{0.41\linewidth}@{}}
\toprule
Application and task & Binding and authority & Intervention/continuity probe and limitation \\
\midrule
Collaboration: share meeting materials & Exact views, recipients, release-policy version; organizer, document owners, and disclosure approver. & Remove a recipient after preview. Rebind approval; an admitted delivery can still commit after cancellation. A received disclosure is not undone by deleting a sent item. \\
Spreadsheet: clean a shared data range & Workbook version, stable row identities, formulas/protected ranges; task requester and workspace policy. & The user sorts or edits while the agent prepares a patch. Reconcile by row identity and pre-values; undo must not overwrite a later collaborator edit. Positional cell coordinates alone are insufficient. \\
IDE: apply a refactoring & Worktree/branch, file versions, proposed diff, command scope; developer rights and repository policy. & The developer edits a file or takes over while a remote worker runs. Fence that worker and reconcile its patch/processes. Local edit approval does not authorize merge, deployment, or arbitrary shell effects. \\
Service console: refund a payment & Payment/order state, amount, beneficiary, policy threshold; operator, account/resource authority, finance approver, affected customer. & Transfer the case to another employee after a timeout. Recheck the recipient employee's rights and query the stable refund ID. Without provider duplicate prevention, blind replay risks a second refund. \\
\bottomrule
\end{tabular}
\end{table}

The differences are substantial. Spreadsheet correctness depends on stable row identity and interference with later edits; an IDE must distinguish a local patch from authority to run commands or deploy; a refund depends on financial policy and the provider's duplicate-prevention semantics. Thus $\mathsf{Post}$ and the assumptions in $\mathsf{Dep}$ are domain-specific, even when task revision and endorsement have a common representation. In each case, the abstraction must preserve consequential intermediate effects as well as the final state. A workflow that briefly discloses confidential information and then deletes it is not equivalent to one that never disclosed it.

The H1--H4 distinction follows the same logic. An editable spreadsheet continuation supports operational takeover, while transferring a service case changes accountable principals. Remote IDE continuation requires durable task and process identity; recovery after a refund timeout requires provider evidence. A common record enables analysis across these cases, but their authority and continuation rules remain different. This is the intended generality of the framework: shared questions and semantic relations with explicit domain interpretations.

\subsection{Human Authority and the Cost of Supervision}
\label{sec:discussion-control}

An inspectable abstraction does not ensure that people can supervise it effectively. A long, technically complete preview can conceal the one changed recipient that matters. The choice of mandatory review must therefore have an accountable source. Application and domain owners define required classes, \eg cross-boundary disclosure, financial commitment, or destructive changes to shared resources. Resource policies add further conditions. A requester may narrow delegation or demand additional review, but cannot waive another principal's restriction. A model can flag uncertainty; it cannot independently downgrade a required gate.

Reversibility, affected people, disclosure scope, commitment cost, unfamiliar authority, and unresolved intent provide candidate dimensions for these decisions. Their interpretation depends on the domain. Formatting a private draft and refunding a payment need different thresholds, while ambiguity about what the user wants differs from uncertainty about a retrievable fact. Human--AI interaction guidance supports correction and usable scoping~\cite{amershi2019guidelines}; the framework locates the corresponding decisions in the task and authority model. When a required classification cannot be established, preparation can continue only within the remaining non-effectful scope.

Supervisory burden is consequently multidimensional. Active inspection time, intervention frequency, context switches, and recovery effort capture different costs; latency and computation are separate measures. Their interpretation requires task outcomes and decision quality, including consequential mismatches approved without detection and interruptions unnecessary under the applicable policy. Reducing the number of confirmation dialogs alone can improve convenience while weakening control. Future human studies should test whether the \toolname makes meaningful differences easier to recognize and act upon for a specified user population.

\subsection{Architectural Scope and Theoretical Limits}
\label{sec:discussion-boundary}
\label{sec:discussion-related}
\label{sec:discussion-models}

\aisname and \toolname operate at different levels of description. The former identifies a software pattern with a continuing core and dual interaction paths; the latter specifies the task semantics through which delegated activity becomes inspectable and controllable. Neither implies a single software module. Their overlap with mixed-initiative or intent-oriented systems does not diminish the need to define the application-level relation, but it does place the burden of contribution on the explanatory and engineering value of that relation. GUI-agent engineering already motivates attention to safety, recovery, and maintenance~\cite{yu2026guiengineering}; the framework gives those concerns a shared semantic object within the application.

The formalization makes selected obligations precise while leaving three limits visible. First, the abstraction must be adequate: a relevant effect omitted from $Y$ cannot be assessed through its traces. Second, its realization must expose trustworthy bindings and evidence; stronger models cannot reconstruct unavailable commit receipts or authoritative permission changes. Third, safety conformance does not imply progress, semantic usefulness, fairness, or usability. Those properties require further assumptions and evaluation. The propositions establish conditional consequences of the definitions and gate assumptions, not verified properties of a deployed system.

These limits also determine adoption choices. An application with inaccessible third-party state or unmediated write paths may support a restricted contract. A deterministic workflow may offer lower cost for a stable goal and procedure. The strongest role for \aisname is therefore not universal replacement of conventional interaction, but selective delegation whose assumptions and consequences remain available to the people who depend on the application.

\subsection{What Would Establish the Value of This Perspective?}
\label{sec:discussion-limitations}

The central hypothesis is that maintaining interaction--effect obligations as shared artifacts can expose defects and dependencies that isolated UI, agent, and API descriptions leave implicit. Its evaluation should compare the effort and outcomes of building and evolving representative integrations with and without those artifacts. Relevant evidence would concern missed binding changes, authority violations, recovery errors, evidence invalidation, and the quality of user interventions. The existing examples identify such observations without supplying empirical results.

The hypothesis could fail in several informative ways. Existing application contracts may already maintain the same relationships at comparable effort. Annotation and evidence costs may outweigh the defects prevented. Users may overlook consequential changes despite a more faithful abstraction. Such findings would narrow the useful scope of the framework and guide lighter realizations. The present perspective provides definitions, conditional arguments, and an agenda for investigating these possibilities; it does not claim an implemented runtime, demonstrated productivity gains, or cross-domain validation.

\section{Conclusion}
\label{sec:conclusion}

Integrating an intelligent agent into an existing application introduces an enduring relationship between direct operations, delegated tasks, and shared business state. We propose \aisname as the software pattern that contains this relationship and \toolname as its task-level interaction abstraction. Modeling their correspondence makes a central obligation explicit: task revisions, role-specific authority, operational control, and outcome evidence must remain connected throughout execution.

Interaction contracts specify that obligation, while continuous assurance maintains the standing of claims about its realization as dependencies change. The resulting research agenda joins application analysis and development assistance with control semantics, quality assessment, and human supervision. Progress should be judged by whether this shared abstraction helps developers build and maintain useful delegation at acceptable cost, while preserving the user's ability to understand and influence what the application does.

% Draft provenance: AI assistance was used for literature discovery, manuscript
% drafting, original conceptual diagrams, and LaTeX checks. Author review and any venue-specific disclosure
% remain necessary before submission. No new empirical study is reported.
% Funding, conflicts of interest, and contributor-role declarations are
% author-owned information and have not been inferred or supplied here.

%\begin{acks}
%The authors would like to thank the anonymous reviewers for their insightful comments.
%This work is partially supported by 
%the National Key Research and Development Program of China ().
%the National Natural Science Foundation of China ().
%\end{acks}

\bibliographystyle{ACM-Reference-Format}
\bibliography{main}

\end{document}